\documentclass[journal,twoside,web,hidelinks]{ieeecolor}
\usepackage{lcsys}
\usepackage{geometry}
\usepackage{url,hyperref}
\usepackage{graphicx}
\usepackage{amsmath,amssymb,amsfonts}
\usepackage{breqn}
\usepackage{microtype}
\usepackage{tensor}
\usepackage{braket}
\usepackage{subcaption}
\usepackage{array}
\usepackage{graphicx}
\usepackage{multirow}
\usepackage{multicol}
\usepackage{color,colortbl}
\usepackage{soul}
\usepackage{comment}
\usepackage{float}
\usepackage{cancel}
\usepackage{siunitx}
\usepackage{flushend}

\allowdisplaybreaks
\usepackage[switch, mathlines]{lineno}
\newtheorem{theorem}{Theorem}
\newtheorem{corollary}{Corollary}

\newtheorem{prop}{Proposition}

\DeclareMathOperator{\tr}{\mathrm{Tr}}

\DeclareMathOperator{\sinc}{\mathrm{sinc}}

\title{Robustness Analysis for Quantum Systems Controlled by Continuous-Time Pulses}
\author{S.\,P.\ O'Neil$^{1,*}$, \and
E.\,A.\ Jonckheere$^2$, \and S.\ Schirmer$^3$
\thanks{$^1$ Department of Electrical Engineering \& Computer Science, United States Military Academy, NY, USA. {\tt sean.oneil@westpoint.edu}}
\thanks{$^2$ Dept of Electrical \& Computer Engineering, University of Southern California, CA, USA. 
{\tt jonckhee@usc.edu}}
\thanks{$^3$ Faculty of Science \& Engineering, Physics, Swansea University, UK. {\tt s.m.shermer@gmail.com}}
\thanks{Any opinions in this work are solely those of the authors and do not reflect those of the US Army, the United States Military Academy, or the US Department of Defense.
}
}
\date{}
\begin{document}
\maketitle
\thispagestyle{empty}

\begin{abstract}
Differential sensitivity techniques originally developed to study the robustness of energy landscape controllers are generalized to the important case of closed quantum systems subject to continuously varying controls.  Vanishing sensitivity to parameter variation is shown to coincide with perfect fidelity, as was the case for time-invariant controls. Upper bounds on the magnitude of the differential sensitivity to any parameter variation are derived based simply on knowledge of the system Hamiltonian and the maximum size of the control inputs.
\end{abstract}


\section{Introduction} 
\IEEEPARstart{G}{iven} the demands for high-precision in proposed applications of quantum technology, the development of control schemes resistant to parameter uncertainty and external noise is paramount to realizing the benefits of a ``second quantum revolution". It is therefore no surprise that techniques for the synthesis of ``robust" quantum controllers abound in the literature. However, these works differ in their approaches based on the employed standard of robustness. The methodology most similar to classical robust control is that of~\cite{James_2007,wang_2023} implementing $H_\infty$ or Linear Quadratic Gaussian (LQG) methods for linear quantum systems modeled by quantum stochastic differential equations.  However, while these techniques provide guaranteed performance margins, they are restricted to linear quantum systems, mainly optical systems, that undergo continuous measurement. 

In~\cite{Ding2025}, the authors presents a novel approach to robust controller design for open quantum systems. However, the approach targets suppression of the noise inherent in the system-bath interaction and is thus not directly applicable to robustness to parameter uncertainty in closed systems. More germane to closed systems, \cite{Daems_2013} and \cite{Van_Damme_2017} define robustness as minimizing the effect of undesirable parameters on the optimal trajectory of a two-level system to first order, invoking the Pontryagin maximum principle to analytically derive the optimal controls resulting from a constrained optimization problem.  While highly effective in simulation, the requirement to explicitly solve for robust pulses in closed form does not scale well beyond two- or three-level systems. 
 
We contribute to the ongoing development of robust quantum control by presenting a procedure to quantify robustness for continuously-varying, smooth controls and applicable to multi-level systems.  Although piecewise constant controls have become a dominant paradigm in quantum control, such controls yield ``solutions" of the Schr\"odinger equation that are piecewise defined and continuous but not differentiable at the switching times. Although further work is needed to explore the advantages of continuously-varying and smooth controls, there are theoretical and practical reasons why smooth controls may be preferable for certain applications. Extending our framework to evaluate robustness to such controls is therefore important. 
As a preview, in Section~\ref{sec: prelims_time_varying}, we generalize the sensitivity analysis of~\cite{o'neil_2024_dynamic_gate_control} and~\cite{o'neil_2023_sensitivity_bounds} to accommodate dynamics governed by continuously varying controls.  As part of this analysis, we introduce the Magnus expansion as a critical tool in the analysis of systems controlled by pulses that depart from the piecewise constant paradigm. After providing conditions for vanishing sensitivity (optimal robustness) in Section~\ref{sec: vanishing_sensitivity}, we develop a bound for the sensitivity to any physically-realizable and norm-bounded uncertainty in Section~\ref{sec: sensitivity_bound}. Finally, Section~\ref{sec: examples} presents two quantum technology-inspired case studies, based on state transfer in a spintronic network and the implementation of a quantum NOT gate.

\section{Preliminaries}\label{sec: prelims_time_varying}

\subsection{System Dynamics and the Magnus Expansion}

We consider closed quantum systems with Hilbert space of finite dimension $N$ whose evolution is given by the time-dependent Schr\"odinger equation 
\begin{equation}\label{eq: schrodinger_equation}
   \imath \hbar \dot{\mathcal{U}}(t) = H(t) \mathcal{U}(t), \quad \mathcal{U}(0) = I.
\end{equation}
$\mathcal{U}(t)$ is an $N \times N$ unitary matrix, $I$ is the identity operator on $\mathbb{C}^N$, and $H(t)$ is a time-dependent Hermitian $N \times N$ matrix.  We restrict attention to Hamiltonians that are a linear combination of a time-invariant drift term $H_0$ and constant control matrices $\set{H_m}$ for $m \in \set{1,\hdots,M}$ modulated by scalar fields $u_m(t)$, so that the Hamiltonian takes the form 
\begin{equation}\label{eq: hamiltonian}
    H(t) = H_0 + \sum_{m=1}^M u_m(t) H_m. 
\end{equation}
For simplicity of notation, we set $A(t):=-\frac{\imath}{\hslash} H(t) = B + \sum_{m=1}^M u_m(t)C_m$ with $B := -\frac{\imath}{\hslash} H_0$ and $C_m := -\frac{\imath}{\hslash} H_m$.

If the controls $u_m(t)$ in \eqref{eq: hamiltonian} are time-invariant then the exact solution of \eqref{eq: schrodinger_equation} is given by the matrix exponential $U(t)=\exp(t A)$.  In practice, constant controls are generally not sufficient to find solutions for most control problems.  One way this problem can be dealt with is by allowing piecewise constant controls, e.g., by sub-dividing the time interval $[0,t_f]$ into subintervals $I_k = [t_k,t_{k+1})$ during which the controls are constant.  This increases the degrees of freedom of the controls while still allowing exact solutions to be constructed from products of matrix exponentials.  This method has proved powerful and become extremely popular in recent years, but it has certain drawbacks.  As the dynamical generator \eqref{eq: hamiltonian} is discontinuous, the resulting piecewise constructed solutions of \eqref{eq: schrodinger_equation} are continuous but not differentiable at the switching times $t_k$.  For many applications this may be acceptable but for some applications it may be desirable for the solutions to be smooth, which necessitates controls that are at least continuous and perhaps differentiable. Unfortunately, as soon as the controls are no longer (piecewise) constant, the solution of \eqref{eq: schrodinger_equation} becomes more complicated as the straightforward generalization of the matrix exponential, $\mathcal{U}(t) = \exp\left(\int_{0}^t A(s) ds \right)$, only applies if the matrices in $\set{B,C_1,\hdots,C_M}$ all commute, a case that is usually not interesting from a control perspective.  In general, more sophisticated techniques are necessary to construct approximate solutions.

\subsection{Approximate Solutions Given by the Magnus Expansion}

A method with the desirable property that (approximate) solutions of \eqref{eq: schrodinger_equation} remain in the unitary group is the Magnus expansion~\cite{magnus_1952_magnus_expansion} 
\begin{equation}
    \mathcal{U}(t) = \exp\left( \sum_{j=1}^\infty \Omega_j(t) \right),
\end{equation}
where the $\Omega_j(t)$ are given by nested commutators and integrals: $\Omega_1(t) = \int_{0}^{t} \! A(t_1) dt_1$, $\Omega_2(t) = \tfrac{1}{2} \int_0^t \int_0^{t_1} [A(t_1),A(t_2)]d t_2 d t_1$, $\Omega_3 =\tfrac{1}{6} \int_0^t \int_0^{t_1} \int_0^{t_2}( [A(t_1), [A(t_2),A(t_3)]] + [A(t_3),[A(t_2), A(t_1)]]) \, d t_3 d t_2 d t_1 $, and more complex expressions as $k$ increases as detailed in~\cite{blanes_2009_magnus_expansion_and_applications, blanes_2010_magus_pedagogical}. To ensure convergence of the Magnus expansion~\cite{casas_2007_magnus_convergence}, the integration domain must be subdivided, and a step size $h$ must be chosen such that on each sub-interval $[s,s+h)$ for $0\le s \le t_f-h$, 
 \begin{equation}\label{eq: step_size_constraint}
     \int_{s}^{s+h} \| A(t) \|dt < \pi,
 \end{equation}
where $\| \cdot \|$ is the spectral norm of the matrix $A(t)$.  Variable step sizes may be used, but we choose a constant step size $h$ that meets the condition in~\eqref{eq: step_size_constraint} based on the norm of the nominal Hamiltonian matrices and bounds on the control amplitudes and the desired measurement time $t_f$, ensuring that $t_f = \tau h$, where $\tau$ is an integer.  The approximate solution $U(t_f)$ can then be expressed as the ordered product
\begin{equation}\label{eq: product_solution}
    U(t_f) = U^{(\tau,\tau-1)} \hdots U^{(1,0)} = \prod_{k=1}^\tau U^{(k,k-1)} =: U^{(\tau,0)} 
\end{equation}
with $t_k = k h$ for $k \in \set{0,1,\hdots \tau}$. This is similar to the solution for piecewise constant controls, and the general approximation for the propagator over the interval $[t_{\tau_1},t_{\tau_2})$ is
\begin{equation*}
    U^{(\tau_2,\tau_1)} := U^{(\tau_2,\tau_2 -1)} \hdots U^{(\tau_1+1,\tau_1)} = \prod_{k = \tau_1+1}^{\tau_2} U^{(k,k-1)}.
\end{equation*}
Given the complexity of the higher order terms in the Magnus expansion and intractability of providing closed-form expressions for the $\Omega_j(t)$, typical application of the Magnus approach requires truncation of the series at a given order.  Numerous methods of various orders have been derived \cite{blanes_2009_magnus_expansion_and_applications}.   For the numerical examples in Section \ref{sec: examples}, we adopt a fourth-order method to approximate the propagator over the time interval $[t_{k-1},t_k)$ for each $k$ from $1$ to $\tau$ as $U^{(k,k-1)} := \exp(\Gamma^{(k)})$. Following~\cite{casas_2007_magnus_convergence,rasulov_2023_beyond_piecewise} the approximate dynamical generator is defined as $\Gamma^{(k)}:=$
\begin{equation}
    \frac{h}{6} \left(A^{(k-1)} + 4A^{(k-1/2)} + A^{(k)} \right) - \frac{h^2}{12} [A^{(k-1)},A^{(k)}],
\end{equation}
where $A^{(k)}:=A(t_k)$, ${A}^{(k+1/2)}:=A(t_k+h/2)$.  As this method is fourth-order, the error of the exact propagator denoted $\mathcal{U}^{(k,k-1)}$ over the time-step $(k-1)h$ to $kh$ compared to the truncated estimate, $ \| \mathcal{U}^{(k,k-1)} - U^{(k,k-1)} \|$, is $\mathcal{O}(h^5)$ \cite{blanes_2009_magnus_expansion_and_applications, casas_2007_magnus_convergence, casas_2006, Iserles1999, Iserles1999_on_implementation_of_magnus}.

\subsection{Fidelity and Performance}

The time-varying controls $\set{u_m(t)}_{m=1}^M$ are synthesized to maximize a performance metric, usually the expectation of an observable or the probability that the system will be found in a given target state.  We generally subsume these into a fidelity measure that is a function of the unitary propagator at the time $t_f$ and denoted as $\mathsf{F}[U(t_f)]$. A typical fidelity measure for state transfer to a given target state is
\begin{equation}\label{eq: state_fideltiy}
    \mathsf{F}_S[U(t_f)] = |\bra{\psi_f} U(t_f) \ket{\psi_0} |^2
\end{equation}
where $\ket{\psi_0}$ is the initial state and $\ket{\psi_f}$ the target state.  For gate operation a typical performance measure is the gate fidelity~\cite{machnes_2011_comparing_benchmarking}
\begin{equation}\label{eq: gate_fidelity}
    \mathsf{F}_G[U(t_f)] = \frac{1}{N} \left| \tr\left( U_f^\dagger U(t_f) \right) \right|
\end{equation}
where $U_f$ is the target unitary process and $\tr$ denotes the matrix trace. In what follows, we omit explicit dependence of $U(t_f)$ on the measurement time $t_f$ with the understanding that $U$ denotes the nominal approximate propagator of~\eqref{eq: product_solution} at time $t_f$.  We also excise the dependence of the fidelity on $U(t_f)$ so that $\mathsf{F}_{\alpha}$ with $\alpha \in\set{S,G}$ denotes the nominal fidelity of~\eqref{eq: state_fideltiy} for states $\{S\}$ or~\eqref{eq: gate_fidelity} for gates $\{G\}$. 

\subsection{First-Order Sensitivity of the Performance}\label{ssec: differential_sensitivity}

Consider the variation in the performance measure $\mathsf{F}_{\alpha}$ due to an uncertainty in one of the matrices $\set{B,C_m}$ modeled as $\delta_\mu S_\mu$ where $\delta_\mu \in \mathbb{R}$ is a scalar uncertainty weight and $S_\mu \in \mathbb{\mathbb{C}}^{N \times N}$ is a skew-Hermitian matrix that encodes the structure of the uncertainty. Denoting the uncertain realization of $A(t)$ by $\tilde{A}(t)$ and the corresponding uncertain propagator by $\tilde{U}(t_f)$, the change in the fidelity due to an uncertainty indexed by $\mu$ is
\begin{equation}
    \mathsf{\tilde{F}} = \mathsf{F} + \left( \partial_\mu \mathsf{\tilde{F}}  \right) \delta_\mu + \mathcal{O}(\delta_\mu^2)
\end{equation}
where $\partial_\mu$ denotes the derivative of $\mathsf{\tilde{F}}$ with respect to the amplitude of the structured perturbation $\delta_\mu S_\mu$. As we wish to analyze robustness to parametric uncertainty, we focus on the first order sensitivity of the fidelity with respect to the structure $S_\mu$ given by $\partial_\mu \tilde{\mathsf{F}}$. The product solution of~\eqref{eq: product_solution} leads to
\begin{equation}\label{eq: state_sensitivity}
    \partial_\mu \tilde{\mathsf{F}}_S  = \sum_{k=1}^\tau 2 \Re\left\{ \bra{\psi_f} X_\mu^{(k)} \ket{\psi_0} \bra{\psi_0} U^\dagger \ket{\psi_f} \right\}
\end{equation}
for the state transfer case and 
\begin{equation}\label{eq: gate_sensitivity}
    \partial_\mu \tilde{\mathsf{F}}_G = \frac{1}{N^2 \mathsf{F}_G}\sum_{k=1}^\tau  \Re \left\{ \tr\left( U_f^\dagger X_\mu^{(k)} \right) \tr \left( U^\dagger U_f \right)  \right\}
\end{equation}
for the gate fidelity problem. In both cases, 
\begin{equation}\label{eq: X_k}
    X_\mu^{(k)} = U^{(\tau,k)} \left( \partial_\mu \tilde{U}^{(k,k-1)} \right) U^{(k-1,0)}.
\end{equation}
The derivative term is given by~\cite{najfeld_1995_dexpma,floether_2012} $\partial_\mu U^{(k,k-1)} =$
\begin{equation}\label{eq: dexmpa} 
     \int_{0}^1 \exp \left( \Gamma^{(k)} (1-s) \right) \left( \partial_\mu \Gamma^{(k)} \right) \exp \left( \Gamma^{(k)} s\right) ds. 
\end{equation}
The above holds for any uncertainty structured as $S_\mu$ that retains the skew-Hermitian property of $A(t)$.  However, for concreteness we explicitly consider collective uncertainty to one of the operators in $\set{B,C_m}$, which captures the two main cases of the uncertainty affecting the system evolution and the control interaction, respectively.  Setting $S_0 = B/\| B\|_F$ and $S_\mu = C_m/\| C_m\|_F$ for $(m=\mu)$ running from $1$ to $M$ and where $\| \cdot \|_F$ is the Frobenius norm, the uncertain, time-varying dynamics matrix $A(t)$ is now
\begin{equation}
    \tilde{A}(t) = (B +  \delta_0 S_0)+ \sum_{m=1}^M u_m(t) C_m 
\end{equation}
for the case of collective uncertainty in $H_0$ or
\begin{equation}\label{eq:S_mu_def}
    \tilde{A}(t) = B + \left(\sum_{m=1}^M u_m(t) C_m + \delta_\mu u_\mu(t)S_\mu\right)
\end{equation}
for the case of collective uncertainty in interaction Hamiltonian $H_\mu$.   As we restrict the analysis to collective uncertainty, $\partial_\mu \Gamma^{(k)}$ takes two distinct forms:
\begin{equation}\label{eq: drift_uncertainty}
    \partial_0 \Gamma^{(k)} = hS_0 + \frac{h^2}{12} \sum_{m=1}^M \left( (u_m^{(k)} - u_m^{(k-1)}) [S_0,C_m] \right)
\end{equation}
for drift uncertainty and 
\begin{multline}\label{eq: control_uncertainty}
    \partial_\mu \Gamma^{(k)} = \frac{h}{6} \left(u_\mu^{(k)} +4u^{(k-1/2)}_{\mu} + u^{(k-1)}_{\mu} \right) S_\mu \\ + \frac{h^2}{12} \left( (u^{(k)}_{\mu} - u^{(k-1)}_{\mu}) [S_\mu,B] \right) - \frac{h^2}{12} \sum_{m \neq \mu} J^{(k)}_{(\mu,m)}[S_\mu,C_m]
\end{multline}
for uncertainty in one of the interaction operators. Here we define $u_m^{(k)}:= u_m(kh)$, and $J^{(k)}_{(\mu,m)}$ is given by $J^{(k)}_{(\mu,m)}:= u_{\mu}^{(k-1)}u_{m}^{(k)} - u^{(k)}_{\mu}u^{(k-1)}_{m}$.  As each matrix in $\set{B,C_m}$ and the associated normalized structures $S_\mu$ are skew-Hermitian, each matrix $\partial_\mu \Gamma^{(k)}$ is in the Lie algebra $\mathfrak{u}(N)$ as the algebra is closed under commutation. Considering the fidelity error, or infidelity, defined by $1-\tilde{\mathsf{F}}_\alpha$, the first-order variation in the error is the sensitivity in~\eqref{eq: state_sensitivity} or~\eqref{eq: gate_sensitivity} multiplied by $-1$. 

\section{Vanishing Sensitivity}\label{sec: vanishing_sensitivity}

As with time-invariant controls, we can provide sufficient conditions for vanishing sensitivity for both state transfer and gate implementation cases.
\begin{theorem}{\cite[Th. 21.7]{Jonckheere1997}}\label{theorem: vanishing}
If the fidelity map $\mathsf{F}: U(N) \rightarrow [0,1]$ is differentiable except possibly at some isolated points, then $\forall U \in \mathsf{F}^{-1}(1)$ where the differential exists, $d_U\mathsf{F}=0$. 
\end{theorem}
Locally charting $U(N)$ with $\mathfrak{u}(N)$, the Lie algebra of skew-Hermitian matrices, using the exponential map, and viewing the physically relevant skew-hermitian $S_\mu$ as providing the direction along which the sensitivity is computed, this sensitivity is $d_{{A}} \mathsf{F} \circ \mathrm{exp}(S_\mu)$   
as seen from~\eqref{eq:S_mu_def}. 
\begin{corollary}{\cite[Th. 3]{schirmer_2017_feedback_laws_spintronics}}\label{theorem: vanishing_state_gate_sensitivity}
If a controller induces perfect state transfer in the sense that the nominal fidelity either $\mathsf{F}_S = 1$ or $\mathsf{F}_G= 1$, the differential sensitivity (first-order variation of the fidelity) with respect to any physically realizable, and hence skew-Hermitian, uncertainty $S_\mu$ is zero. 
\end{corollary}

\section{Bounding the Differential Sensitivity to Arbitrary Perturbations}\label{sec: sensitivity_bound}

At the other end of the spectrum, we consider upper bounds on the potential size of the differential sensitivity for a given controller for a set of allowable uncertainties or perturbations to the dynamics. To better analyze the relationship between a set of possible perturbation structures, represented by the family $\set{\partial_\mu \Gamma ^{(k)}}$, and the resulting sensitivity, we seek an expression such that the perturbation structure encoded in $\partial_\mu \Gamma^{(k)}$ may be viewed as the input to a controller-specific operator that yields the resulting sensitivity to a given structure. Central to this derivation is the controller-specific mapping
\begin{align*}
\mathcal{D}^{(k)}: \mathfrak{u}(N) \rightarrow \mathfrak{u}(N), && \partial_\mu \Gamma^{(k)} & \mapsto {U^{(k,k-1)}}^\dagger \partial_\mu U^{(k,k-1)}
\end{align*}
where the action of $\mathcal{D}^{(k)}$ on $ \partial_\mu \Gamma^{(k)}$ is explicitly given by $\int_{0}^1 \exp \left( -s \Gamma^{(k)} \right) \left( \partial_\mu \Gamma^{(k)} \right) \exp \left( s\Gamma^{(k)} \right) ds$. The integrand is the adjoint representation of the $U(N)$-group element $\exp\left( -s\Gamma^{(k)}\right)$ acting on the element $\partial_\mu \Gamma^{(k)}$ of its $\mathfrak{u}(N)$-Lie algebra.   To assist in the following manipulations, we define $\mathsf{D}^{(k)}$ as the matrix representation of the $\mathcal{D}^{(k)}$ super-operator, obtained by expanding $\partial_\mu \Gamma^{(k)}$ in a basis $\set{\sigma_n}_{n=1}^{N^2}$ for $\mathfrak{u}(N)$. 

Before proceeding, note that for the case of state transfer we have $\braket{\psi_f | X_\mu^{(k)} | \psi_0} = \tr \left( \mathcal{D}^{(k)} \left( \partial_\mu \Gamma^{(k)} \right) P^{(k)} \right)$ with $P^{(k)} = U^{(k-1,0)} \ket{\psi_0} \bra{\psi_f} U^{(\tau,k-1)} := \ket{\alpha^{(k)}} \bra{\beta^{(k)}}$ for $k>1$ and $P^{(1)}:= \ket{\psi_0}\bra{\psi_f}U$. Likewise, for gate operation we have that $\tr\left( U_f^\dagger X_\mu^{(k)}\right) = \tr\left( \mathcal{D}^{(k)} \left( \partial_\mu \Gamma^{(k)}\right) P^{(k)} \right)$ with $P^{(k)} = U^{(k-1,0)} U_f^\dagger U^{(\tau,k-1)}$. We may then unify the expression of the sensitivity of the fidelity for both cases as 
\begin{equation}\label{eq: sensitivity_state_compact}
    \partial_\mu \tilde{\mathsf{F}}_{\alpha} = \gamma \sum_{k = 1}^\tau \Re \left\{ \tr \left( \mathcal{D}^{(k)} \left( \partial_\mu \Gamma^{(k)} \right) P^{(k)} z^* \right) \right\}
\end{equation}
where $z^* = \bra{\psi_0} U^\dagger \ket{\psi_f}$ and $\gamma = 2$ for state transfer (see~\eqref{eq: state_sensitivity}) and $z^* = \tr(U^\dagger U_f)$ and $\gamma=1/(N^2 \mathsf{F}_G)$ for the gate fidelity problem (see~\eqref{eq: gate_sensitivity}). We then have the following proposition. 
\begin{prop}
    The Hermitian part of $P^{(k)}z^*$ makes no contribution to the sensitivity of the fidelity. 
\end{prop}
\begin{proof}
    For brevity write $\mathcal{D}^{(k)} \left( \partial_\mu \Gamma^{(k)}\right)$ as $X$ and $P^{(k)}z^*$ as $Y$. By the properties of the adjoint action of a Lie group element on its Lie algebra, $X$ is in $\mathfrak{u}(N)$ and so skew-Hermitian. Decompose $Y$ as $Y = Y_H + Y_{SH}$ where $Y_H$ is Hermitian and $Y_{SH}$ is skew-Hermitian. It then follows that $\Re\left\{ \tr \left( XY\right) \right\} = \Re \left\{ \tr \left(  X Y_{SH} \right) \right\}$ as $\Re \left\{ \tr \left( X Y_H\right) \right\}=0$.  
\end{proof}
It then suffices to expand $\left(P^{(k)}z^*\right)_{SH}$ in the same basis of $\mathfrak{u}(N)$. Specifically we have $\left(P^{(k)}z^*\right)_{SH} = \sum_{n=1}^{N^2} \mathsf{p}_n^{(k)} \sigma_n$ where $\mathsf{p}_n^{(k)} = \Re \left\{ \tr \left( P^{(k)} z^* \sigma_n^{\dagger} \right) \right\}$. Similarly, we expand $\partial_\mu \Gamma^{(k)}$ as $\sum_{n=1}^{N^2} \left( \mathsf{g}_\mu^{(k)} \right)_n \sigma_n$. We then represent $P^{(k)}z^*$ and $\partial_\mu \Gamma^{(k)}$ by the real $N^2$ vectors consisting of the expansion coefficients. To generate a matrix representation of the superoperator $\mathcal{D}^{(k)}$ we have the $N^2 \times N^2$ matrix $\mathsf{D}^{(k)}$ with elements given as 
\begin{equation}
    \mathsf{D}^{(k)}_{m n} = \int_0^1 \left( \tr \left( e^{-\Gamma^{(k)}s} \sigma_n e^{\Gamma^{(k)}s} \sigma_m^\dagger \right) \right)ds.
\end{equation}

Taking the spectral decomposition of $\Gamma^{(k)}$ as $\sum_{p=1}^{N} \Pi^{(k)}_p i \lambda^{(k)}_p$, where $i\lambda^{(k)}_p$ is the $p$th eigenvalue of $\Gamma^{(k)}$ with associated one-dimensional projector $\Pi^{(k)}_p$, let $\omega^{(k)}_{pq} := \lambda^{(k)}_p - \lambda^{(k)}_q$.  Then
\begin{multline}
 \int_0^1 e^{i(\lambda_p -\lambda_q)s} \, ds = \left. \frac{e^{i (\lambda_p - \lambda_q)s}}{i (\lambda_p - \lambda_q)} \right|_0^1 = \frac{e^{i\omega_{pq}}-1}{i\omega_{pq}} \\ = e^{i\omega_{pq}/2} \frac{e^{i\omega_{pq}/2}-e^{-i\omega_{pq}/2}}{i\omega_{pq}} = e^{i\omega_{pq}/2}\sinc(\omega_{pq}/2)
\end{multline}
shows that the elements of the matrix representation of the $\mathcal{D}^{(k)}$ super-operator may be also written as 
\begin{equation*}
    \mathsf{D}^{(k)}_{mn} = \sum_{p,q = 1}^{N} \tr \left( \Pi^{(k)}_q \sigma_n \Pi^{(k)}_p \sigma_m^\dagger \right) e^{i \omega^{(k)}_{pq}/2} \sinc \left( \frac{\omega^{(k)}_{pq}}{2} \right). 
\end{equation*}
In either case, we now write the expression for the sensitivity of the fidelity error as the real matrix-vector product 
\begin{equation}
    \partial_\mu \tilde{\mathsf{F}}_{\alpha} = \gamma \sum_{k=1}^\tau {\mathsf{p}^{(k)}}^T \mathsf{D}^{(k)} \mathsf{g}_\mu^{(k)} = \gamma \sum_{k=1}^{\tau} \mathsf{z}^{(k)} \mathsf{g}_\mu^{(k)}
\end{equation}
where $\mathsf{z}^{(k)}$ is an $N^2$-dimensional, real row vector given by ${\mathsf{p}^{(k)}}^T \mathsf{D}^{(k)}$. Noting that $\| \mathsf{g}_\mu^{(k)} \| = \| \partial_\mu \Gamma^{(k)}\|_F$ we can then bound the sensitivity with respect to any norm bounded perturbation structure $\| \partial_\mu \Gamma^{(k) }\|_F \leq b$ in accordance with the theorem below.
\begin{theorem}\label{theorem: upper_bound_on_sensitivity}
For a norm bounded perturbation such that $\| \partial_\mu \Gamma^{(k)} \|_F = \| \mathsf{g}_\mu^{(k)}\| \leq b$ for all time steps $k$, the magnitude of the differential sensitivity to this perturbation structure is bounded above by $     \left| \partial_\mu \tilde{\mathsf{F}}_{\alpha} \right| \leq \gamma \| \mathbf{Z} \|_1 b =: \beta$ where $\mathbf{Z}$ is the $N^2$-dimensional vector with $k$th component given as $\mathsf{Z}_k = \| \mathsf{z}^{(k)}\| = \|{\mathsf{p}^{(k)}}^T \mathsf{D}^{(k)} \|$. 
\end{theorem}
\begin{proof}
The proof follows from the exposition above, which gives $\partial_\mu \tilde{\mathsf{F}}_{\alpha} = \gamma \sum_{k=1}^\tau \mathsf{z}^{(k)} \mathsf{g}_\mu^{(k)}$. Taking absolute values on both sides of the equation and employing the triangle equality to the sum yields $| \partial_\mu \mathsf{F}_{\alpha} | \leq \gamma \sum_{k=1}^\tau \left| \mathsf{z}^{(k)} \mathsf{g}_\mu^{(k)} \right|$. Applying Cauchy-Schwarz to the inner product of $\mathsf{z}^{(k)}$ with $\mathsf{g}_\mu^{(k)}$ for each $k$ and enforcing the bound on $ \| \mathsf{g}_\mu^{(k)}\| \leq b$ we have $\left| \partial_\mu \tilde{\mathsf{F}}_{\alpha} \right| \leq \gamma \left( \sum_{k=1}^\tau \| \mathsf{z}^{(k)} \| \right) b = \gamma \| \mathbf{Z}\|_1 b $.
\end{proof}

This construction allows us to consider the controller and the input-output state information encoded in the sequence $\mathsf{z}^{(k)}$ separately from the perturbation-specific effects encoded in $\partial_\mu \Gamma^{(k)}$, and efficiently evaluate the worst-case sensitivity of a given controller for a collection of perturbations $\set{\partial_\mu \Gamma^{(k)}}_\mu$. Instead of bounding the sensitivity of the fidelity for a given controller with regard to certain perturbations, one can alternatively derive lower bounds on the fidelity considering a range of perturbations, an approach pursued in~\cite[Th. 1]{Kosut2025}.  Both approaches can be seen as complementary.

\section{Applications}~\label{sec: examples}

We now apply results of Section~\ref{sec: sensitivity_bound} to two different quantum technology platforms subject to distinct time-varying control protocols. We do this to illustrate that the results are not dependent upon a specific application or control pulses synthesized from a specific set of basis functions. 

\subsection{State Transfer in Spin Ring}\label{ssec: state_transfer}

We first consider four coupled spin-$1/2$ particles with ring topology as introduced in~\cite{schirmer_2017_feedback_laws_spintronics}. But instead of designing controllers based on static fields, we permit time-varying control fields.  Restricting the dynamics to the single-excitation subspace, and defining $E_{mn}$ as the $4 \times 4$ matrix with a $1$ in the $(m,n)$th position and $0$ elsewhere, the drift Hamiltonian has the explicit form $H_0 = 
\sum_{n=1}^3 \left( E_{n,n+1} + E_{n+1,n} \right) + E_{4,1} + E_{1,4}$.  We choose two control Hamiltonians. $H_1 = E_{1,1}$ and $H_2 = E_{2,2}$, which suffices to ensure full controllability in accordance with the Lie algebra rank condition~\cite{d'alessandro_intro_quantum_control}. The control objective is to maximize the state transfer fidelity of~\eqref{eq: state_fideltiy} with initial state $\ket{\psi_0} = e_1$ and target $\ket{\psi_f} = e_4$ where $\set{e_k}_{k=1}^4$ are the natural basis vectors for $\mathbb{C}^4$.  

There are many bases for functions defined on $[0,t_f]$ that generate acceptable controllers, but for illustration, we choose to build the scalar control fields $u_1(t)$ and $u_2(t)$ from a basis set of cubic polynomial splines~\cite{mcclarren_2018_cubic_splines}.  Dividing the time interval $[0,t_f]$ into $K$ equal intervals, we assume controls of the form
\begin{equation*}
    u_m(x) = \sum_{n=1}^K \left( a^{(n)}_m x^3 + b^{(n)}_m x^2 + c^{(n)}_m x + d^{(n)}_m \right) I(n)
\end{equation*}
where $0 \le x \le t_f/K$ for each $n$ and $I(n)=1$ for $(n-1)t_f/K \le t \le nt_f/K$ and zero otherwise. The vector of parameters for the optimization is the $4KM$-dimensional vector $\mathbf{x} = \left(a^{(1)}_1,b^{(1)}_1,c^{(1)}_1,d^{(1)}_1,\hdots,a^{(K)}_M,b^{(K)}_M,c^{(K)}_M,d^{(K)}_M \right)$.  We set $d^{(1)}_m = 0$ for all $m$ to ensure the initial value of $u_m(t)$ is zero and place appropriate constraints on the coefficients at the boundary of each spline to ensure the resulting field is $\mathcal{C}^1$.  Based on trial and error in varying the read-out time $t_f$, number of splines $K$, and bound the on the controls $|u_m(t)|$, we choose $t_f = 5$, $K = 3$, and $|u_m(t)| \leq 5$ which generates a selection of controllers with $\mathsf{F}_S > 0.99$.  
The controls are the solution of minimizing $\mathsf{e}_S = 1-\mathsf{F}_S$ as per equation~\eqref{eq: state_fideltiy} using a constrained quasi-Newton optimization.  In units with $\hbar=1$, $\|B \| = \| -i/\hslash H_0 \| = 2$ and $\| C_m \| = \| -i/ \hslash H_m \| = 1$. Given the bound $\mathcal{B}=5$ on the control amplitudes, we employ~\eqref{eq: step_size_constraint} to obtain a very conservative estimate of $\kappa=50$ for the integration step size that ensures convergence. This corresponds to a time step of $h = 1/30$ and $\tau = K \times \kappa = 150$ steps on the interval $[0,t_f]$. 

\begin{figure}[t]
\centering
{\includegraphics[width = 0.85\columnwidth]{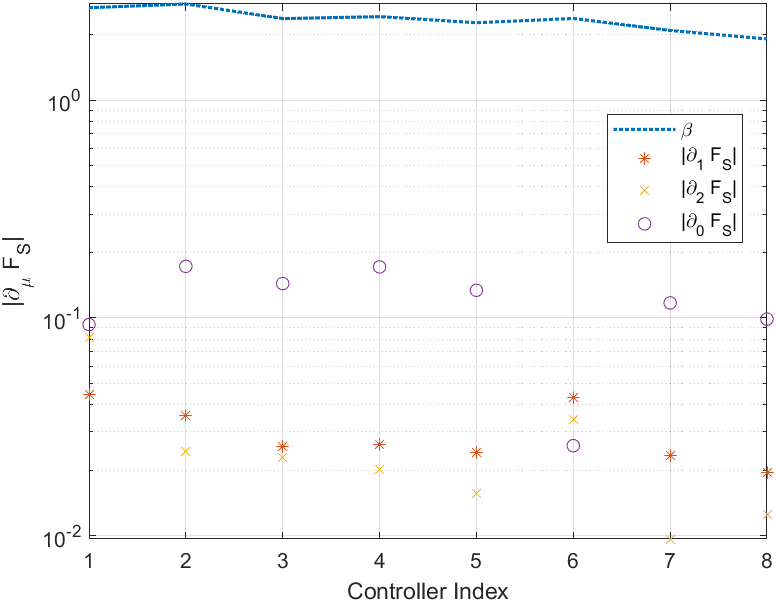}}
\caption{Plot comparing $|\partial_\mu \mathsf{F}_S |$ for $\mu=\set{0,1,2}$ with $\beta$ from Theorem~\ref{theorem: upper_bound_on_sensitivity}. Though computed only from the nominal system and control data, the bound has an accuracy to within slightly more than one order of magnitude. The controllers are ordered by increasing fidelity so that the controller with index 1 yields the smallest fidelity and that with index 8 yields the greatest.}   
\label{fig: N=4_bound}
\end{figure} 
Given a set of high fidelity controllers, the differential sensitivity to collective uncertainty in any of the Hamiltonian matrices is readily computed by application of~\eqref{eq: state_sensitivity} and~\eqref{eq: X_k} to~\eqref{eq: control_uncertainty}. Based on the bound $\mathcal{B}$ used in the synthesis, we compute the maximum norm of the perturbations $\partial_\mu \Gamma^{(k)}$ as $b = \max\limits_{\mu,k} \| \partial_\mu \Gamma^{(k)} \| = 0.1555$. We then test the bound $\beta$ of Theorem~\ref{theorem: upper_bound_on_sensitivity} after computing $\| \mathbf{Z}\|_1$ as derived in Section~\ref{sec: sensitivity_bound}.  Figure~\ref{fig: N=4_bound} shows the result.  Without the need to explicitly compute the derivative of a matrix exponential as in~\eqref{eq: dexmpa}, and based simply on knowledge of the controller and norm of the matrices composing the Hamiltonian, the bound on the worst-case sensitivity is accurate to nearly an order of magnitude. Specifically, for the data in Figure~\ref{fig: N=4_bound}, the mean of the ratio of $\beta$ to $\max\limits_\mu \left| \partial_\mu \mathsf{F}_S\right|$ is $23.1$. 

\subsection{Gate Operation}

To demonstrate that the analysis is applicable beyond the state transfer problem, we consider a quantum gate implementation for the three-level, super-conducting circuit model presented in~\cite{pagano_2024_bases}.  In units with $\hbar = 1$, the time-varying dynamical generator takes the form   
\begin{multline*}
    A(t) = B + u_1(t)C_1 + u_2(t)C_2 = \\ 
    \begin{bmatrix} 0 & 0 & 0 \\ 0 
 & 0 & 0 \\ 0 & 0 & -i \Delta \end{bmatrix} - u_1(t) \begin{bmatrix} 0 & i & 0 \\ i & 0 & i \\ 0 & i & 0 \end{bmatrix} + u_2(t) \begin{bmatrix} 0 & -1 & 0 \\ 1 & 0 & -1 \\ 0 & 1 & 0 \end{bmatrix}
\end{multline*}
where $\Delta$ is the strength of the \textit{anharmonicity} that determines how much the $\ket{2} \rightarrow \ket{3}$ transition differs from the $\ket{1} \rightarrow \ket{2}$ transition.  We use the same target gate as in~\cite{pagano_2024_bases}, a NOT gate with explicit representation 
 \(   U_f = \begin{bmatrix} 0 & 1 & 0 \\ 1 & 0 & 0 \\ 0 & 0 & 1 \end{bmatrix}. \)
To demonstrate applicability to another choice of basis functions, we construct the controls $u_1(t)$ and $u_2(t)$ from a restricted Fourier basis. To ensure the control fields have a realistic initial condition of $u_i(0) = 0$, we cast them as 
\begin{equation}\label{eq: fourier_basis}
    u_m(t) = \sum_{j = 1}^{\mathsf{J}} \alpha_{mj} \sin(\omega_{mj} t)
\end{equation}
and optimize over the $4\mathsf{J}$-dimensional vector of control parameters $\mathbf{x} = \left( \alpha_{m1},\hdots,\alpha_{m\mathsf{J}},\omega_{m1},\hdots,\omega_{m\mathsf{J}} \right)$ for $m=1,2$. After trial and error in varying $t_f$ and the bound on $\set{\alpha_{mj}}$ and $\set{\omega_{mj}}$, we choose $t_f = 2$, $| \alpha_{mj} | \leq 2$ and $0\le \omega_{mj} \leq 10$, which generated a selection of acceptable-fidelity ($\mathsf{F}_G > 0.99$) controllers by minimizing the objective $\mathsf{e}_G = 1-\mathsf{F}_G$ following from~\eqref{eq: gate_fidelity} and employing the same constrained quasi-Newton optimization as before. A step size analysis similar to that of Section~\ref{ssec: state_transfer} yields $\tau = \kappa = 50$ that guarantees convergence. 

\begin{figure}[t]
\centering
{\includegraphics[width = 0.85\columnwidth]{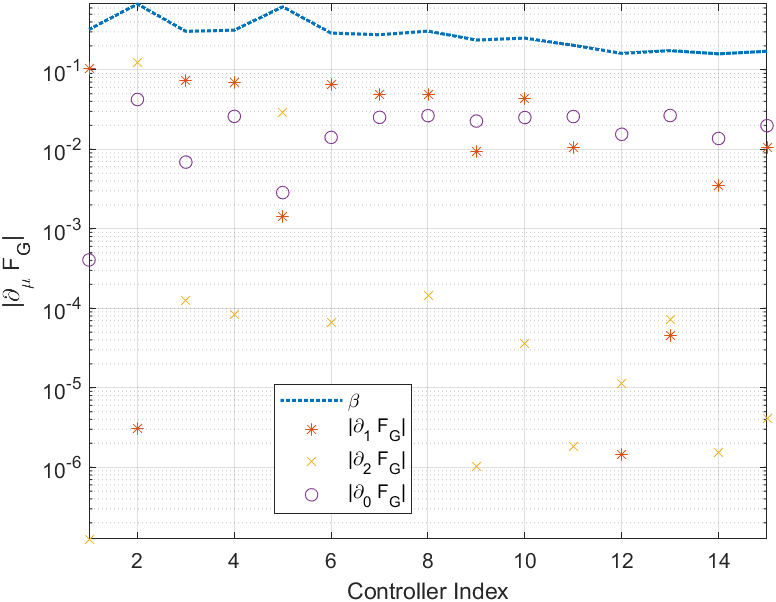}}
\caption{Plot of $\beta$ as given in Theorem~\ref{theorem: upper_bound_on_sensitivity} for the gate synthesis case with $|\partial_\mu \mathsf{F}_S |$ for $\mu=\set{0,1,2}$. The bound $\beta$ predicts the largest sensitivity to collective uncertainty in any of the Hamiltonian matrices to within an order of magnitude. The controllers are ordered by increasing fidelity so that the controller with index 1 yields the smallest fidelity and that with index 15 yields the greatest.} 
\label{fig: N=3_bound}
\end{figure}

Similar to the first example, we establish an upper bound on the norm of allowable perturbations such that $\| \partial_\mu \Gamma^{(k)} \| \leq b = 0.1287$.  We then compare the bound given by $\beta = \gamma \| \mathbf{Z} \|_1 b$ with the absolute value of the differential sensitivity computed from Section~\ref{ssec: differential_sensitivity}. As seen in Figure~\ref{fig: N=3_bound}, the bound is accurate to within an order of magnitude for all controllers with fidelity error between $10^{-2}$ and $10^{-3}$. Specifically the ratio of the bound to $\max\limits_\mu | \partial_\mu \mathsf{F}_G |$ has a mean of $7.73$ for this set of controllers.

\section{Conclusion}\label{sec: conclusion_time_varying}

We extended the applicability of the differential sensitivity as an analysis tool to systems driven by continuously varying controls.  We showed that the worst-case sensitivity resulting from a set of norm-bounded perturbations can be reliably bounded with only knowledge of the nominal system and controller.  We also showed that vanishing sensitivity to Hamiltonian uncertainty, first observed for state transfer induced by static fields in~\cite[Th. 3]{schirmer_2017_feedback_laws_spintronics}, is more general, extending to gate operation and time-varying controls.

Future work should focus on \textit{synthesis} of controls with reduced sensitivity.  We conjecture that minimization of $ \| {\mathbf{Z}} \|_1$ yields improved robustness, as defined by the size of the differential sensitivity, to arbitrary uncertainty or perturbation structures, although this requires further exploration.  Another challenge for continuous controls requiring further exploration are accurate and efficient numerical integration schemes with guaranteed bounds on the numerical approximation errors, perhaps beyond the Magnus expansion. Finally, the applying the vanishing sensitivity results to \textit{open} systems should be investigated to better determine the generality of this property. 

\bibliographystyle{ieeetr}
\bibliography{bibliography}

\end{document}